\documentclass[12pt]{article}
\pdfoutput=1

\usepackage{amsmath,amssymb,calc}
\usepackage{amsthm}
\usepackage{thmtools}
\usepackage{mathrsfs}
\usepackage{bm}
\usepackage{bbm}
\usepackage{mathtools}
\usepackage{enumerate}
\usepackage[authoryear]{natbib}
\usepackage{caption,subcaption}
\usepackage{pgfplots}
\usepackage{chngcntr}
\usepackage{appendix}
\usepackage[colorlinks,citecolor=blue,urlcolor=blue]{hyperref}
\usepackage{cleveref}
\usepackage{crossreftools}
\usepackage{titlesec}
\usepackage{hyphenat}
\usepackage{booktabs}

\usepackage{authblk}
\author[1]{Tetsuya Kaji}
\author[2]{Jianfei Cao}
\affil[1]{University of Chicago}
\affil[2]{Northeastern University}

\title{Assessing Heterogeneity of Treatment Effects%
\thanks{This work is supported by the Liew Family Junior Faculty Fellowship and the Richard N.\ Rosett Faculty Fellowship at the University of Chicago Booth School of Business and the support for visiting faculty at the Department of Economics at the Massachusetts Institute of Technology.
We thank Yijing Zhang for excellent research assistance. We thank Rachael Meager, Toru Kitagawa, Kosuke Imai, Alberto Abadie, Whitney Newey, Ismael Mourifi\'e, Eric Auerbach, Ivan Canay, Federico Bugni, Filip Obradovic, Muyang Ren, Peter Hull, Josh Angrist, Marinho Bertanha, Simon Lee, and seminar participants at various places for helpful comments.}}
\date{\today}

\allowdisplaybreaks[3]

\theoremstyle{plain}
\newtheorem{thm}{Theorem}
\newtheorem*{thm*}{Theorem}

\theoremstyle{definition}

\theoremstyle{remark}

\renewcommand\thmcontinues[1]{continued}

\usepackage{sectsty}
\subsubsectionfont{\normalfont\itshape}

\newcommand{\exasymbol}{$\square$}

\declaretheoremstyle[
spaceabove=6pt, spacebelow=6pt,
headfont=\normalfont\bfseries,
notefont=\mdseries, notebraces={(}{)},
bodyfont=\normalfont,
postheadspace=1em,
qed=\exasymbol
]{exa}
\declaretheorem[style=exa,name=Example]{exa}

\declaretheorem[style=exa,name=Example,unnumbered]{exa*}

\counterwithin{subsection}{section}
\counterwithin{subsubsection}{subsection}

\crefname{thm}{Theorem}{Theorems}
\crefname{prop}{Proposition}{Propositions}
\crefname{lem}{Lemma}{Lemmas}
\crefname{exa}{Example}{Examples}
\crefname{section}{Section}{Sections}
\crefname{subsection}{Section}{Sections}
\crefname{subsubsection}{Section}{Sections}
\crefname{appendix}{Appendix}{Appendices}

\crefrangelabelformat{exa}{#3#1#4--#5\crefstripprefix{#1}{#2}#6}

\makeatletter
\renewcommand*\env@matrix[1][*\c@MaxMatrixCols c]{%
  \hskip -\arraycolsep
  \let\@ifnextchar\new@ifnextchar
  \array{#1}}
\makeatother

\usepackage[margin=1.25in]{geometry}
\usepackage[onehalfspacing,nodisplayskipstretch]{setspace}

\makeatletter
\newcommand{\subalign}[1]{%
  \vcenter{%
    \Let@ \restore@math@cr \default@tag
    \baselineskip\fontdimen10 \scriptfont\tw@
    \advance\baselineskip\fontdimen12 \scriptfont\tw@
    \lineskip\thr@@\fontdimen8 \scriptfont\thr@@
    \lineskiplimit\lineskip
    \ialign{\hfil$\m@th\scriptstyle##$&$\m@th\scriptstyle{}##$\hfil\crcr
      #1\crcr
    }%
  }%
}
\makeatother

\usepackage{accents}

\begin{document}

\maketitle

\begin{abstract}
Heterogeneous treatment effects are of major interest in economics. For example, a poverty reduction measure would be best evaluated by its effects on those who would be poor in the absence of the treatment, or by the share among the poor who would increase their earnings because of the treatment. While these quantities are not identified, we derive nonparametrically sharp bounds using only the marginal distributions of the control and treated outcomes. Applications to microfinance and welfare reform demonstrate their utility even when the average treatment effects are not significant and when economic theory makes opposite predictions between heterogeneous individuals.

\end{abstract}

\section{Introduction}

Many questions in economics involve knowledge about heterogeneous treatment effects across different outcome levels.
Consider, for example, the evaluation of microfinance as a poverty reduction measure.
Suppose that the outcome of interest is wealth.
If the average treatment effect (ATE) is positive, the access to microfinance increases individuals' wealth on average.
However, if this is brought about by rich people being richer and poor people being poorer, the policymaker may want to think twice before rushing to implementation.
In another example, consider evaluating a welfare reform on low\hyp{}income workers.
The effect of the reform can be complex in that the anticipated income effect and price effect might point at different directions, even under simple and stylized economic theory.
An important piece of information is how many workers choose to work more under the new welfare program.

This paper develops new methods to help answer these questions.
To fix ideas, let $Y_{i1}$ be subject $i$'s treated outcome and $Y_{i0}$ subject $i$'s control outcome.
Two quantities are of interest here.
The first is the {\em subgroup treatment effect (STE)},
\[
	\mathbb{E}[Y_{i1}-Y_{i0}\mid Y_{i0}<c],
\]
where $c$ is some number.
In the first example, $Y$ corresponds to the wealth of individual $i$, so the STE represents the ATE of the subgroup defined by low wealth.
The second quantity is the {\em subgroup proportion of winners (SPW)},
\[
	P(Y_{i1}>Y_{i0}\mid Y_{i0}<c).
\]
In the second example, $Y$ corresponds to the earnings from work for individual $i$, so the SPW gives the share of low\hyp{}earning individuals who earn more under the new welfare program.%
\footnote{If paid by hours, the hours worked are proportional to earnings. Then, we might regard higher earnings as more hours worked.}

Both quantities, however, require knowledge of individual treatment effects (ITEs) because of the conditioning on $Y_{i0}$.
Since ITEs are not identified, neither quantity is point\hyp{}identified.
However, the marginal distributions of $Y_{i0}$ and $Y_{i1}$ put restrictions on the maximum ranges in which they lie.
Consequently, we seek bounds that are as tight as possible with no additional assumptions.

Our bounds provide useful complementary information to the widely\hyp{}used quantile treatment effect (QTE).
The QTE is often used to display the degree of heterogeneity immanent in the data but overlooked by the ATE \citep[e.g.,][]{bgh2006}.
However, the QTE does not come with a handy interpretation as the ITE unless one is willing to assume rank invariance; hence it calls for a delicate interpretation to say anything more than the mere existence of heterogeneity.
Our bounds are available almost whenever QTEs are, and they characterize the furthest extent to which the data can speak about the heterogeneity in the forms of STEs and SPWs.

The bounds depend only on the marginal distributions of the treated outcome and the control outcome, so they are applicable widely beyond the randomized controlled trial (RCT) framework.%
\footnote{We do require, however, that the treatment be binary.}
For example, there are existing estimators of the potential outcome distributions in the binary IV models with monotone compliance, selection\hyp{}on\hyp{}observables models, panel data, and regression discontinuity designs.
These estimators were mainly proposed as intermediate quantities to obtain the QTE, but instead of (or in addition to) calculating the QTE, we can feed them into our bounds to derive additional, interpretable insights on heterogeneity.
The Stata program to compute the bounds from the RCT data (with monotone compliance) is available at an \href{https://github.com/jcao0/subgroup-treatment-effects}{author's website}.

This paper adds to the enormous literature on partial identification of heterogeneity measures in program evaluation.
The literature is too vast for us to give a comprehensive review, but
some papers close to ours include \citet{hsc1997}, \citet{fp2010}, and \citet{t2012}.
\citet{hsc1997} discussed bounds on the proportion of winners and the distribution of treatment effects and investigated how the assumption of rational choice by the participants helps reduce uncertainty.
\citet{fp2010} gave bounds on the distribution of treatment effects in terms of second\hyp{}order stochastic dominance.
\citet{t2012} derived bounds on positive and absolute treatment effects and the proportion of winners on the entire population.
In contrast, we consider subgroups that are defined by the ranges of control outcome, giving a finer picture on the distribution of the treatment effects.
\citet{jls2023} considered robust estimation of partially identified causal effects exploiting information of covariates.
When we wish to incorporate covariates into estimation, their method provides a convenient way to produce reliable estimators of the bounds.

Our STE bounds are also related to the Lee bounds on the treatment effect under sample selection \citep{hm2000,l2009a}.
Suppose $Y_i$ is wage, $D_i$ is eligibility to a job training program, and $S_i$ the indicator of $i$'s employment status.
Since the wage is only observed when employed, the outcome is missing when $S_i=0$.
Denote by $S_{id}$ the potential indicator of employment when $D_i=d$.
They concern bounding the effect of the job training program on wages for individuals who would be employed under both treatment statuses, $\mathbb{E}[Y_{i1}-Y_{i0}\mid S_{i0}=1,S_{i1}=1]$.
Under the standard assumptions of the literature, this boils down to bounding $\mathbb{E}[Y_{i1}\mid S_{i0}=1,S_{i1}=0]$ when the distribution of $Y_{i1}$ given $S_{i1}=1$ and $P(S_{i0}=1\mid S_{i1}=1)$ are identified.
We have an analogous situation: our objective is to bound $\mathbb{E}[Y_{i1}-Y_{i0}\mid Y_{i0}<c]$ when the marginal distribution of $Y_{i1}$ and $P(Y_{i0}<c)$ are identified.

The rest of the paper is organized as follows. 
\cref{sec:motivating} motivates our bounds with several examples in development economics, health economics, and labor economics.
\cref{sec:results} lays out the formal statements of our results.
\cref{sec:micro} presents an application of the STE to microfinance, drawn from \citet{tdj2015}.
\cref{sec:labor} gives an application of the SPW to the evaluation of a welfare reform, borrowing from \citet{bgh2006}.
\cref{sec:concl} concludes.
All proofs are contained in \cref{sec:theory}.

\section{Motivating Examples} \label{sec:motivating}

We discuss a few stylized examples to motivate each of the STE and SPW.
In all examples, and throughout the paper, we maintain the potential outcomes notation: $Y_{i1}$ for the treated and $Y_{i0}$ for the control outcomes.

\subsection{Subgroup Treatment Effects}

\begin{exa}[Microfinance for the poor]
Microfinance allows poor households to borrow a small amount of money, with which they can invest in crops or livestock to escape from a poverty trap.
The outcome of interest is therefore a measure of wealth or income, such as the value of livestock owned or the net revenues from crops.
Let $Y_i$ be the value of livestock owned by $i$.
On top of the ATE, $\mathbb{E}[Y_{i1}-Y_{i0}]$, we are interested specifically in the ATE for those who would have had low values of livestock in possession,
\[
	\mathbb{E}[Y_{i1}-Y_{i0}\mid Y_{i0}<Y_0^b],
\]
where $Y_0^b$ is the $b$th quantile of $Y_{i0}$.
It is of interest to look at this value at various values of $b\in(0,1)$, and if it is positive for most low values of $b$, we might conclude that microfinance helps the poor.

In \cref{sec:micro}, we revisit this example using the RCT data from \citet{tdj2015}.
They gave estimates of the ATE%
\footnote{They considered access to microfinance as the treatment. If one sees borrowing from microfinance as the treatment, the ATE estimates can be understood as the intent\hyp{}to\hyp{}treat (ITT) estimates.}
and stated
\begin{quote}
Despite the large increase in borrowing, we find that for a large majority of socioeconomic outcomes the null of no impact cannot be rejected, although in several cases the point estimates are substantially large but imprecisely estimated.
\end{quote}
When we apply our bounds for the outcome of livestock value, we find that the STE for the poor is significantly positive, and the impreciseness of the estimates is most likely due to the large variation of the livestock values of the upper tail.
\end{exa}

\begin{exa}[Healthcare for the unhealthy]
Many unhealthy individuals are not covered by health insurance, and it is of interest to know whether they would become healthier if they were covered.
Therefore, the treatment of interest is the access to healthcare, and the outcome of interest is a measure of health such as the blood pressure.
Let $Y_i$ be the blood pressure and $[c_1,c_2]$ be the ideal range thereof.
One quantity of interest is the effect of access to healthcare on the blood pressure for those with otherwise low blood pressure,
\[
	\mathbb{E}[Y_{i1}-Y_{i0}\mid Y_{i0}<c_1].
\]
If this is positive, availability of healthcare improves the health of those with low blood pressure.
We may also be interested in the effect on the blood pressure for those with otherwise high blood pressure,
\[
	\mathbb{E}[Y_{i1}-Y_{i0}\mid Y_{i0}>c_2].
\]
If this is negative, healthcare is beneficial for those as well.
As the direction of benefits changes depending on where the control outcome is, such health benefits are difficult to capture by ATEs.
\end{exa}

\subsection{Subgroup Proportions of Winners}

\begin{exa}[Evaluation of welfare reform]
When replacing an old welfare program with a new one, Connecticut conducted a randomized experiment to evaluate its effect.
In essence, this reform replaced a mild indefinite support for women with a more generous but time\hyp{}limited support.
\citet{bgh2006} observed that theory predicts heterogeneous effects regarding the sign and magnitude of the response of labor supply.

Individual $i$'s income can be decomposed into earnings and government transfers.
Let $Y_i$ be $i$'s earnings.
As we will see in \cref{sec:labor}, stylized theory predicts positive effects on earnings for low\hyp{}earnings workers before the time limit, so if both leisure and consumption are normal goods for a good amount of workers, then we expect
\[
	P(Y_{i1}>Y_{i0}\mid Y_{i0}<c)
\]
to be positive for a low value of $c$.
After presenting QTEs, \citet{bgh2006} stated
\begin{quote}
Without further assumptions, the possibility of rank reversals prevents us from being more specific about who the winners and losers are.
\end{quote}
In \cref{sec:labor}, we apply our bounds on this example and find that the lower bound is significantly positive.
\end{exa}

\begin{exa}[Do-no-harm principle]
The principle of ``do no harm'' is considered to be a primary ethical standard of the medical profession.
Let $Y_i$ be a biomarker that measures the level of health. Assume for simplicity that the larger the better in the observed range.
Then, the treatment is beneficial to $i$ when $Y_{i1}>Y_{i0}$ (a ``winner'' from the treatment), while it is harmful to $i$ when $Y_{i1}<Y_{i0}$ (a ``loser'' from the treatment).
If there are losers among those with low levels of $Y_{i0}$, that is, $Y_{i0}<c$, that means the treatment exacerbates the conditions of those who are already bad.
Thus, it is of interest to know the magnitude of
\[
	P(Y_{i1}<Y_{i0}\mid Y_{i0}<c)
\]
in deciding a suitable treatment.
\end{exa}

\begin{exa}[Persuasion effect]
When the provision of information is the treatment, of interest is whether the information affects the decision of the information receiver.
This is called the {\em persuasion effect}.
\citet{dk2007} considered the persuasion effect of Fox News on voting for the Republican candidate in Presidential elections.
To be precise, let $Y_i$ be a binary indicator of whether voter $i$ votes for the Republican candidate.
The treatment is the availability of Fox News in $i$'s region.
Persuasion by Fox News is defined as $Y_{i0}=0$ and $Y_{i1}=1$, that is, in the absence of Fox News, $i$ votes for the Democrat candidate, but with Fox News, for the Republican.
This corresponds to the ``winner'' in our paper.
The quantity of interest is the {\em persuasion rate},
\[
	P(Y_{i1}>Y_{i0}\mid Y_{i0}=0).
\]
\citet{jl2023} showed that the persuasion rate is partially identified in empirically relevant setups and derived sharp bounds thereof.
In fact, our bounds (\cref{thm:spw}) reduce to equation (3) of their paper when $Y$ is binary.

Now, the results of this paper can be considered a natural extension of the persuasion effect to continuous decisions, e.g., whether the medical recommendation to sleep longer induces patients to sleep longer and how such effects differ across different levels of initial hours of sleep.
\end{exa}

\section{Main Results} \label{sec:results}

\subsection{Setting}

We do not directly observe the potential outcomes $Y_{i0}$ and $Y_{i1}$.
Instead, we observe the treatment indicator $D_i$ and the corresponding outcome $Y_i=D_i Y_{i1}+(1-D_i)Y_{i0}$.
We let $F_0$ and $F_1$ denote the marginal cumulative distribution functions (cdfs) of $Y_{i0}$ and $Y_{i1}$ for the population of interest.

Depending on the application, $F_0$ and $F_1$ for the intended population may be identified through different mechanisms.
Here, we list some examples.
Thereafter, we take the identification of $F_0$ and $F_1$ as given.

\begin{itemize}

	\item {\em Case 1 (RCT):} The treatment indicator $D_i$ is exogenous. This arises when the treatment is randomized or when the assignment is randomized but the focus is on the ITT.
	In this case, $F_0$ and $F_1$ are trivially identified and correspond to the cdfs of outcomes for the entire population in the experiment.

	\item {\em Case 2 (Imperfect compliance):} We observe a binary instrument $Z_i$ that satisfies monotonicity. This arises when assignment to the treatment is randomized.
	If one\hyp{}sided compliance holds, $F_0$ and $F_1$ for the compliers are identified as a normalized difference of observable cdfs \citep[eq.~(5--6)]{a2002}.%
	\footnote{The plug\hyp{}in estimator is not guaranteed to be monotonic, but we can rearrange it to be \citep{cfg2009}.}
	This has been extended to cases where $Z_i$ is valid conditional on covariates \citep{aai2003,fm2013,p2020}.

	\item {\em Case 3 (Difference\hyp{}in\hyp{}differences):} We observe cross\hyp{}sectional pre\hyp{}treatment outcomes.%
	\footnote{By ``cross\hyp{}sectional'' we mean that individuals need not be tracked over time.}
	If the ``change\hyp{}in\hyp{}changes'' assumption holds, $F_0$ and $F_1$ for the treated are identified as the composition of observable cdfs and quantiles \citep[Theorem 3.1]{ai2006}.
	This has been generalized to synthetic control by \citet{g2023}.
	He proposed to construct $F_0$ for the treated as a convex combination of those for the control, where the weights are determined as the minimizer of the 2\hyp{}Wasserstein distance in the pre\hyp{}treatment periods.

	\item {\em Case 4 (Selection on observables):} We observe covariates $W_i$ such that $D_i$ is independent of $(Y_{i0},Y_{i1})$ conditional on $W_i$.
	If the common support assumption holds, $F_0$ and $F_1$ for the whole sample and for the treated are identified through propensity score weighting \citep[Lemma 1]{f2007}.

	\item {\em Case 5 (Regression discontinuity design):} We observe the running variable $R_i$ such that the propensity score $P(D_i=1\mid R_i)$ jumps discontinuously at a cutoff $R_i=r_0$. Under the assumption that there are no defiers at the cutoff, $F_0$ and $F_1$ for the compliers at the cutoff are identified, either for sharp or fuzzy regression discontinuity designs \citep[Lemma 1]{ffm2012}.

\end{itemize}

In all cases, we may also observe exogenous covariates $X_i$ and wish to condition our analysis on the values of $X_i$.
Then, we will be concerned of the conditional cdfs of $Y_{i0}$ and $Y_{i1}$ conditional on $X_i$.
Unless $X_i$ is finitely supported, this calls for some modeling of the conditional distribution functions \citep[e.g.,][]{ch2006,aai2003}.
For simplicity, we hereafter suppress $X_i$ in our notation and understand implicitly that all results are allowed be conditioned on $X_i$.

From this point on, we take for granted that $F_0$ and $F_1$ for the population of interest are identified.
Other than that, we do not make any more assumptions; notably, we do not require $Y_{i0}$ or $Y_{i1}$ to be continuous or discrete.

\subsection{Definition of the Subgroup}

Let $Q_0$ and $Q_1$ be the quantile functions corresponding to $F_0$ and $F_1$, that is,
\[
	Q_j(u)=\inf\{x\in\mathbb{R}:F_j(x)\geq u\}
\]
for $0<u<1$ and $j=0,1$.
In plain words, $Q_j$ is the inverse of $F_j$ but assigns the smallest value when there are flat regions of $F_j$.

To define subgroups, we introduce the {\em rank} of individual $i$ in terms of the control outcome, to be denoted by $U_i\in[0,1]$.
This rank is defined so that it satisfies
\[
	Y_{i0}=Q_0(U_i)
\]
for each individual in the population.
If $Y_{i0}$ has a continuous distribution, there is a one\hyp{}to\hyp{}one relationship between $U_i$ and $Y_{i0}$ such that $U_i=F_0(Y_{i0})$.
If $Y_{i0}$ has a discrete mass, there can be many distinct values of $U_i$ that correspond to a mass value of $Y_{i0}$ and we have $U_i\leq F_0(Y_{i0})$ in general.
In this case, we may understand $U_i$ as representing the unobservable heterogeneity of otherwise observationally equivalent individuals in this mass.
It is without loss of generality to assume that $U_i$ is distributed uniformly over $[0,1]$.

Consequently, the subgroup
\(
	\{Y_{i0}<c\}
\)
can also be represented as
\(
	\{U_i<F_0(c)\}
\),
but the subgroup
\(
	\{U_i<b\}
\)
is not necessarily identical to
\(
	\{Y_{i0}<Q_0(b)\}
\).
This means that conditioning on the values of $U_i$ is {\em finer} than conditioning on the values of $Y_{i0}$, that is, every conditioning on $Y_{i0}$ can be represented by some conditioning on $U_i$, but the converse is not true.
For this generality, we focus on the subgroups of the form
\[
	\{a<U_i<b\}
\]
for various values of $0\leq a<b\leq 1$, rather than of the form $\{a<Y_{i0}<b\}$.

The introduction of this slight hassle will pay off when we plot the bounds against the rank.
Conditioning on $U_i$ produces smooth curves that traverse over continuously increasing subpopulations no matter what the distribution of $Y_{i0}$ is.%
\footnote{For example, the Lipschitz property in \cref{thm:winner} is thanks to this formulation.}

\subsection{Subgroup Treatment Effects} \label{sec:ste}

\subsubsection{Statement}

The following theorem gives the sharp lower and upper bounds on the STE for the subgroup defined by $\{a<U_i<b\}$.

\begin{thm}[Bounds on subgroup treatment effects] \label{thm:ste}
For every $0\leq a<b\leq 1$,
\begin{align}
	\mathbb{E}[Y_{i1}-Y_{i0}\mid a<U_i<b]&\geq\frac{1}{b-a}\int_a^b[Q_1(u-a)-Q_0(u)]du,\label{thm:ste:1}\\
	\mathbb{E}[Y_{i1}-Y_{i0}\mid a<U_i<b]&\leq\frac{1}{b-a}\int_a^b[Q_1(1+a-u)-Q_0(u)]du,\label{thm:ste:2}
\end{align}
provided that the conditional expectation exists.
For each of the bounds, there exists a joint distribution of $(Y_{i0},Y_{i1})$ that attains it.
\end{thm}

These bounds are nonparametrically sharp, meaning that there exists a joint distribution of $(Y_{i0},Y_{i1})$, having marginal distributions $F_0$ and $F_1$, that attains each bound.
To name a few intuitive cases, the lower bound (\ref{thm:ste:1}) for $a=0$ is attained when {\em rank invariance} holds, i.e., $Y_{i1}=Q_1(U_i)$; the upper bound (\ref{thm:ste:2}) for $a=0$ is attained when {\em complete rank reversal} holds, that is, $Y_{i1}=Q_1(1-U_i)$.

\begin{figure}
\centering
\includegraphics[page=1]{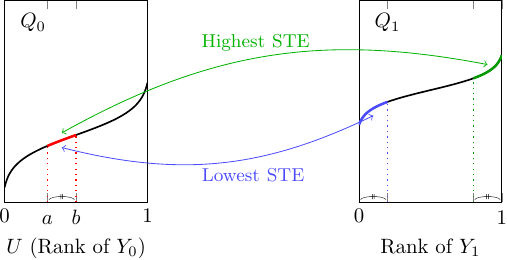}
\caption{Intuition of the STE bounds. The red region on the left indicates the $Y_0$ for the subgroup of interest $\{a<U<b\}$. The lowest $Y_1$ they can receive is the blue region on the right; the highest is the green region on the right.}
\label{fig:ste}
\end{figure}

The intuition of the bounds is visualized in \Cref{fig:ste}.
The left figure plots $Q_0$, and the red region describes our subpopulation of interest, $\{a<U_i<b\}$.
The right figure plots $Q_1$.
While we do not know where in $Q_1$ each individual in the red region corresponds to, we know that the lowest possible ATE these people can receive {\em collectively} is when they are matched with the blue region on $Q_1$.
This gives (\ref{thm:ste:1}).
Note that the exact correspondence between the points in red and those in blue can be left unspecified for this bound to be exact.
Also, the highest possible ATE they can receive is when they are matched with the green region on $Q_1$.
This gives (\ref{thm:ste:2}).
If we let $a=b$, they reduce to the trivial bounds,
\[
	Q_1(0)-Q_0(a)\leq\mathbb{E}[Y_{i1}-Y_{i0}\mid U_i=a]\leq Q_1(1)-Q_0(a).
\]

The above intuition suggests that the lower bound would be the most informative when $a=0$ and the upper bound when $b=1$.
In fact, in the application in \cref{sec:micro}, we fix $a=0$ (since we focus on the poor) and view the bounds as functions of $b$.

While the QTE cannot be interpreted as the ITE in the absence of rank invariance, (\ref{thm:ste:1}) states that the {\em integral} of the QTE can be interpreted as the lower bound for the STE with $a=0$ {\em without requiring} rank invariance.
As stated earlier, this lower bound is exact when rank invariance holds, but it is also expected to be close when rank invariance approximately holds (see also \Cref{fig:steexa:L} in the next section).
The idea that individuals do not move ranks too drastically by the treatment might appear reasonable in economic applications.

The natural estimator of the bounds is the plug\hyp{}in estimator, where the quantile functions are replaced by the estimated ones.
For example, in the ITT analysis in \cref{sec:micro}, we plug in the empirical quantile functions of $Y_i\mid D_i=0$ and of $Y_i\mid D_i=1$ into $Q_0$ and $Q_1$.

If the estimated bounds are asymptotically jointly normal for fixed $a$ and $b$, we may construct an asymptotically valid confidence interval for the STE using the methods developed by \citet{im2004} and \citet{s2009}.
This kind of confidence interval does {\em not} cover the true bounds (an interval) but {\em does} cover the true STE (a single point) with a desired confidence level.
This is also implemented in \cref{sec:micro}.

\subsubsection{Normal Distribution Example} \label{sec:ste:normal}

\begin{figure}[t]
\centering
\begin{subfigure}[t]{0.29\textwidth}
\centering
\includegraphics[page=1]{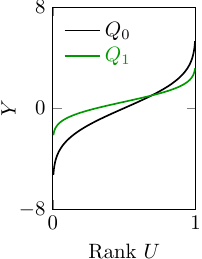}
\caption{Quantile functions.}
\label{fig:steexa:Q}
\end{subfigure}
\ 
\begin{subfigure}[t]{0.29\textwidth}
\centering
\includegraphics[page=2]{fig_STEexa.pdf}
\caption{STE for $\{U<b\}$.}
\label{fig:steexa:L}
\end{subfigure}
\ 
\begin{subfigure}[t]{0.38\textwidth}
\centering
\includegraphics[page=3]{fig_STEexa.pdf}
\caption{STE for $\{U>a\}$.}
\label{fig:steexa:R}
\end{subfigure}
\caption{Illustration of the STE bounds. The left gives an example when marginal distributions are normal. When $(Y_0,Y_1)$ is jointly normal, we can compute the unidentified population STE. The middle figure plots the STE for $\{U<b\}$ and the right for $\{U>a\}$ for various values of correlation as color\hyp{}coded by the color bar. The shaded areas indicate the interiors of our bounds. Perfect correlation cases attain the bounds.}
\label{fig:steexa}
\end{figure}

We illustrate the STE using an example where the outcomes are normally distributed.
Let $Y_{i0}\sim N(0,2)$ and $Y_{i1}\sim N(1/2,1)$, so their quantile functions look like \Cref{fig:steexa:Q}.
If we specify their joint distribution to be normal, we can compute the unidentified population STE.

\Cref{fig:steexa:L} plots the STE,
\[
	\mathbb{E}[Y_{i1}-Y_{i0}\mid U_i<b],
\]
as a function of $b$ for various values of correlation between $Y_{i0}$ and $Y_{i1}$, namely for $\rho=-1$, $-0.5$, $0$, $0.5$, and $1$.
The color represents the value of correlation as given by the color bar on the right.

For every value of $b$, the population $\{U<b\}$ receives the lowest STE when $Y_{i0}$ and $Y_{i1}$ are perfectly correlated, and the highest when they are perfectly negatively correlated.
At any rate, we can infer that the lower group receives the STE higher than the ATE.
The gray shaded area represents the region covered by our bounds.

\Cref{fig:steexa:R} plots the STE for $\mathbb{E}[Y_{i1}-Y_{i0}\mid U_i>a]$.
On the contrary, the population $\{U>a\}$ receives the STE lower than the ATE.
This STE is the highest when the outcomes are perfectly correlated and the lowest when perfectly negatively correlated.

In applications where rank migration is mild, we expect the true STEs to be close to the lower bound for $\{U<b\}$ and the upper bound for $\{U>a\}$.

\subsubsection{Welfare Bounds}

Sometimes the ATE is used for welfare analysis.
We may also derive the bounds on possibly {\em non\hyp{}utilitarian} welfare.
In fact, \cref{thm:ste} is an immediate corollary of the following result.

\begin{thm}[Bounds on subgroup welfare] \label{lem:1}
Let $f:\mathbb{R}\to\mathbb{R}$ be a nondecreasing convex function and $g:\mathbb{R}\to\mathbb{R}$ a nonincreasing convex function.
For every $0\leq a<b\leq 1$,
\begin{align*}
	\mathbb{E}[f(Y_{i1}-Y_{i0})\mathbbm{1}\{a<U_i<b\}]&\geq\int_a^b f(Q_1(u-a)-Q_0(u))du,\\
	\mathbb{E}[f(Y_{i1}-Y_{i0})\mathbbm{1}\{a<U_i<b\}]&\leq\int_a^b f(Q_1(1-u+a)-Q_0(u))du,
\end{align*}
and
\begin{align*}
	\mathbb{E}[g(Y_{i1}-Y_{i0})\mathbbm{1}\{a<U_i<b\}]&\geq\int_a^b g(Q_1(1-b+u)-Q_0(u))du,\\
	\mathbb{E}[g(Y_{i1}-Y_{i0})\mathbbm{1}\{a<U_i<b\}]&\leq\int_a^b g(Q_1(b-u)-Q_0(u))du,
\end{align*}
provided that the expectations exist.
For each of the bounds, there exists a joint distribution of $(Y_0,Y_1)$ that attains it.
\end{thm}

Non\hyp{}utilitarian welfare might arise as a consequence of loss aversion.
When the provider of the treatment may be held responsible for the {\em negative} effects of the treatment, such as doctors \citep{b2009} or policymakers \citep{nh2020}, the provider might be interested in maximizing a non\hyp{}utilitarian welfare.%
\footnote{On a related note, \citet{hsc1997} questioned utilitarian welfare in the context where individuals can choose treatment.}
In \cref{sec:micro}, we use this result to illustrate the policy that maximizes the worst\hyp{}case welfare.

While the claim is intuitive to understand, the proof without a distributional assumption on $Y$ calls for extention of classical optimal transport to conditioning on the subgroup.
We achieve this by alternating the characterizations of the bound and of the optimal transport plan, as detailed for \cref{thm:sosd} in \cref{sec:theory}.

\cref{lem:1} can also be used to bound the positive and negative STEs,
\begin{gather*}
	\mathbb{E}[\max\{Y_{i1}-Y_{i0},0\}\mid a<U_i<b],\\
	\mathbb{E}[\min\{Y_{i1}-Y_{i0},0\}\mid a<U_i<b].
\end{gather*}
These bounds may help policymakers examine the sign of the treatment effects.
Note that the sum of the bounds on positive and negative effects does not give the tightest bound on the net effect as given in \cref{thm:ste}.
Therefore, if we want the bounds on all of the positive, negative, and net effects, they need to be calculated for each case separately.

\subsection{Subgroup Proportions of Winners} \label{sec:spw}

Now we turn to the SPW.
We introduce the notation $[x]_+=\max\{x,0\}$ and $F(a-)=\lim_{x\nearrow a}F(x)$.
The mirror image of the SPW is the {\em subgroup proportion of losers (SPL)},
\[
	P(Y_{i1}<Y_{i0}\mid a<U_i<b).
\]

\subsubsection{Statement}

The following theorem characterizes the sharp lower and upper bounds on the SPW and SPL for the subgroup of interest $\{a<U_i<b\}$.

\begin{thm}[Bounds on subgroup proportions of winners and losers] \label{thm:spw}
For every $0\leq a<b\leq 1$,
\begin{align}
	P(Y_{i1}>Y_{i0}\mid a<U_i<b)&\geq\frac{1}{b-a}\sup_{u\in(a,b)}[u-a-F_1(Q_0(u))]_+,\label{thm:spw:1}\\
	P(Y_{i1}>Y_{i0}\mid a<U_i<b)&\leq 1-\frac{1}{b-a}\sup_{u\in(a,b)}[b-u-1+F_1(Q_0(u))]_+,\notag
\end{align}
and
\begin{align}
	P(Y_{i1}<Y_{i0}\mid a<U_i<b)&\geq\frac{1}{b-a}\sup_{u\in(a,b)}[b-u-1+F_1(Q_0(u)-)]_+,\label{thm:spw:2}\\
	P(Y_{i1}<Y_{i0}\mid a<U_i<b)&\leq1-\frac{1}{b-a}\sup_{u\in(a,b)}[u-a-F_1(Q_0(u)-)]_+.\notag
\end{align}
For each of the lower bounds, there exists a joint distribution of $(Y_0,Y_1)$ that attains it; for each of the upper bounds, there exists a joint distribution of $(Y_0,Y_1)$ that makes it arbitrarily tight.
\end{thm}

Similarly as the STE, these bounds are the most informative when we let either $a=0$ or $b=1$.
In practice, therefore, we recommend fixing $a=0$ and looking at the bounds as functions of $b$, or fixing $b=1$ and looking at the bounds as functions of $a$.

\begin{figure}
\centering
\begin{subfigure}[t]{0.48\textwidth}
\centering
\includegraphics[page=1]{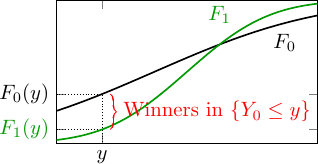}
\caption{For the lower bound of SPW.}
\label{fig:spw:1}
\end{subfigure}
\quad
\begin{subfigure}[t]{0.48\textwidth}
\centering
\includegraphics[page=2]{fig_spw.pdf}
\caption{For the lower bound of SPL.}
\label{fig:spw:2}
\end{subfigure}
\caption{Intuition of the SPW bounds. The black and green lines indicate the cdfs for the control and the treated.}
\label{fig:spw}
\end{figure}

The analytical formulas of the SPW bounds are more complicated than those of the STE bounds, but we may nevertheless attempt to intuit them as in \Cref{fig:spw}.
Since the upper bounds are complements of the lower bounds, it suffices to focus on the lower bounds.
Assume for simplicity that both outcomes are continuous and let $a=0$.
For an arbitrary value of outcome $y$, the portion of individuals whose $Y_0$ is below $y$ is given by $F_0(y)$, and the portion whose $Y_1$ is below $y$ by $F_1(y)$ (\Cref{fig:spw:1}).
If $F_0(y)>F_1(y)$, then at least $F_0(y)-F_1(y)$ portion of individuals in the subgroup $\{Y_0\leq y\}$ must move from $Y_0\leq y$ to $Y_1>y$.
They are obviously winners.
Since this is true for every value of $y$, we can take the supremum in the range of interest.
Redefining $y=Q_0(u)$ gives (\ref{thm:spw:1}).

Next, suppose instead that $F_1(y)>F_0(y)$ as in \Cref{fig:spw:2}.
Then, the difference $F_1(y)-F_0(y)$ gives the sure losers in the group $\{Y_0>y\}$.
Meanwhile, this group contains $1-b$ portion of individuals who do not belong to the subgroup of interest.
The worst scenario is that all of those individuals count toward the sure losers.
Thus, the difference
\[
	[F_1(y)-F_0(y)]-(1-b)
\]
provides the lower bound of losers in the subgroup $\{U<b\}$.
Redefining $y=Q_0(u)$ and taking the supremum in the range produce (\ref{thm:spw:2}).
This intuition may smack as loose, but it turns out to be sharp.
The proof in \cref{sec:theory} contains a construction of the joint distribution that attains this bound.

Note that the conditional probabilities in \cref{thm:spw} are obviously continuous in $(a,b)$ even when $Y_0$ and $Y_1$ are discrete since $U$ is continuously distributed, while the bounds appear susceptible to discontinuities.
Somewhat surprisingly, the bounds are shown to be continuous, making them sharp regardless of the distributions of $Y$.

\cref{thm:spw} can be trivially extended to bound
\[
	P(Y_{i1}-Y_{i0}<c\mid a<U_i<b)
\]
for arbitrary $c$ by shifting the distribution of either $Y_{i0}$ or $Y_{i1}$.
If we then set $a=0$ and $b=1$, the bounds reduce to the classical Makarov bounds \citep{m1981}.
In this sense, \cref{thm:spw} generalizes the Makarov bounds to conditioning on $U_i$.

Note that the bounds in \cref{thm:spw} are in the form of {\em intersection}, that is, the bounds are given by suprema.
Simple plug\hyp{}in estimators for the bounds of this kind can be severely biased in finite samples.
\citet{clr2013} developed a procedure that gives an asymptotically median\hyp{}bias\hyp{}corrected estimator of the bounds as well as an asymptotically valid confidence interval for the partially identified parameter.
The idea is to approximate the supremand function by a Gaussian process and adjust for the precision of its estimator.
The estimation procedure goes as follows.
\begin{enumerate}
	\item Estimate the covariance function of the supremand function.
	\item Simulate the zero\hyp{}mean Gaussian process with the corresponding covariance function and obtain the median of its supremum (over a shrinking set).
	\item Subtract the median times the standard error function from the supremand function.
	\item Maximize the precision\hyp{}corrected supremand function.
\end{enumerate}
In the application in \cref{sec:labor}, we use this method to construct the estimators of our bounds and the pointwise confidence intervals for the SPW.
As before, the confidence interval thusly constructed does not cover the true bounds but covers the true SPW with a specified confidence level.

Finally, as the bounds are attained at somewhat quirky distributions, we may be interested in restricting the set of joint distributions we consider.
In the most general case, we may numerically solve the optimal transport problem for the SPW or SPL with respect to the restricted set, but we may lose the closed\hyp{}form expressions of the bounds.
In \cref{sec:labor:tighten}, we restrict the support of the joint distribution in a tractable way so that we maintain the closed\hyp{}form expressions given by \cref{thm:spw} while still tightening the bounds.

\subsubsection{Normal Distribution Example}

\begin{figure}[t]
\centering
\begin{subfigure}[t]{0.29\textwidth}
\centering
\includegraphics[page=1]{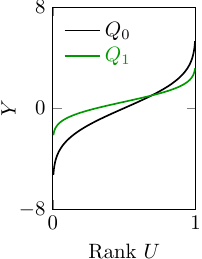}
\caption{Quantile functions.}
\label{fig:spwexa:Q}
\end{subfigure}
\ 
\begin{subfigure}[t]{0.29\textwidth}
\centering
\includegraphics[page=2]{fig_SPWexa.pdf}
\caption{SPW for $\{U<b\}$.}
\label{fig:spwexa:L}
\end{subfigure}
\ 
\begin{subfigure}[t]{0.38\textwidth}
\centering
\includegraphics[page=3]{fig_SPWexa.pdf}
\caption{SPW for $\{U>a\}$.}
\label{fig:spwexa:R}
\end{subfigure}
\caption{Illustration of the SPW bounds. The left gives an example when marginal distributions are normal. When $(Y_0,Y_1)$ is jointly normal, we can compute the unidentified population SPW. The middle figure plots the SPW for $\{U<b\}$ and the right for $\{U>a\}$ for various values of correlation as color\hyp{}coded by the color bar. The shaded areas indicate the interiors of our bounds.}
\label{fig:spwexa}
\end{figure}

We use the same example as \cref{sec:ste:normal} and illustrate the SPW.
The marginal distributions are maintained to be $Y_{i0}\sim N(0,2)$ and $Y_{i1}\sim N(1/2,1)$.
If we impose joint normality, we can compute the unidentified population SPW.

If they are perfectly correlated, roughly the lower 69\% of the population gains from the treatment, and the upper 31\% loses.
Therefore, the true SPW is one until $b=0.69$ and then decreases to 0.69 as we approach $b=1$ (the dark red line in \Cref{fig:spwexa:L}).
Similarly, if the outcomes are perfectly negatively correlated, the lower 57\% gains from the treatment; this is the dark blue line in \Cref{fig:spwexa:L}.
Other correlation values yield nondegenerate distributions, and there are both winners and losers at every rank.

The gray area gives the region between our bounds.
For example, at $b=0.3$, the value of the lower bound is $0.80$.
This means that among the lower 30\% of the population in terms of $Y_{i0}$, at least 80\% of them must be winners, $Y_{i1}>Y_{i0}$.
This drops to 50\% when we expand the subgroup to $b=0.5$.
At $b=0.95$, the lower bound is $0.26$ and the upper bound is $0.96$.
This means that at least 26\% of $\{U_i<0.95\}$ are winners and at least 4\% are losers; the remaining 70\% may be a winner, a loser, or neither.

The losers are easier to identify on the right tail.
In the right figure, at $a=0.9$, the value of the upper bound is $0.20$.
Note that the flipside of a winner is a loser.%
\footnote{In this example, $Y_{i0}=Y_{i1}$ occurs with probability zero.}
Therefore, we can interpret that among the upper 10\% subpopulation, about 80\% of them have to be losers.

Note that these bounds are pointwise tight, but there is usually no joint distribution that attains the bounds uniformly over $a$ or $b$.

\section{Application 1: STE for Microfinance} \label{sec:micro}

We present an empirical application to illustrate how our methods can be used to draw information on the heterogeneous treatment effects in practice.
We estimate the bounds on the effect of access to microfinance on the total value of livestock for the poor.
We borrow the setup and data from \citet{tdj2015}, who analyzed the RCT of microfinance conducted in rural Ethiopia from 2003 to 2006.%
\footnote{The replication files were downloaded at an author's website \citep{tdj2015dta}.}

The brief outline of the experiment is as follows.
Randomization was carried out at the Peasant Association (PA) level, a local administrative unit.
Out of 133 PAs, 34 PAs were randomly assigned to the treatment and 33 PAs to the control.
The remaining 66 PAs were assigned to a different combination of treatments, which we will not use.
The implementation agencies sometimes did not comply with the experimental protocol, and the actual treatment coincided with the assignment in 78\% of the cases.
In this sense, our analysis is on the ITT---the effect of random assignment---following \citet{tdj2015}.
We will call this the ATE throughout this section.

The data we use are from the postintervention survey, but there was also a preintervention survey where individuals were not tracked between the two.
Therefore, another possible direction, which we will not pursue, is to look at the effects on the treated using the framework of \citet{ai2006}.
We leave further details of the experiment to \citet{tdj2015}.

\citet{tdj2015} noted that ``most loans were initiated to fund crop cultivation or animal husbandry, with 80 percent of the 1,388 loans used for working capital or investment in these sectors . . .''
In light of this, we choose the total value of livestock owned as the outcome of interest.

\subsection{Do the Poor Benefit from Microfinance?}

We are interested in knowing how the assignment to microfinance impacts those who are most in need, i.e., those who would attain low outcomes in the absence of the treatment.%
\footnote{\citet{acw2016} stated that ``many researchers and policymakers are interested in estimating how treatments affect those most in need of help, that is, those who would attain unfavorable outcomes in the absence of the treatment.''}

We let $Y_i$ be individual $i$'s total value of livestock owned.
It is in the local currency units, Birr, in their 2006 value.%
\footnote{In January 2006, 100 Birr was roughly worth 11.4 USD.}
We define ``those in need'' as follows.
Sort individuals based on their (unobservable) potential outcome under control, $Y_{i0}$.
Then, take a sequence of cumulative subgroups whose $Y_{i0}$ belongs to the lower tail.
Precisely, using the rank notation, we look at the subgroups $\{U_i<b\}$ as we vary $b$ from $0$ to $1$.

First, the ATE is estimated to be 226.50 Birr, but the 95\% confidence interval contains 0, ranging from $-12.43$ to 469.69.
However, this does not mean that the data have nothing to say about the effects on the low $Y_{i0}$.
\Cref{fig:tdj2015:ste} plots the bounds on the STE characterized by \cref{thm:ste},
\[
	\mathbb{E}[Y_{i1}-Y_{i0}\mid U_i<b],
\]
for various values of $0<b\leq 1$.
The thick black line gives the lower bound of the STE for each $b$, and the dashed line the upper bound.
As mentioned in \cref{sec:ste}, since we set $a=0$ in \cref{thm:ste}, the lower bound gives a meaningful picture but the upper bound is of little value.
As the upper bound becomes unreasonably high, we show only the relevant range.

Note that the STE is point\hyp{}identified and equals the ATE when $b=1$; thus, the two bounds converge to the ATE at $b=1$.
Between $b=0$ to $b=0.169$, the lower bound is equal to zero.
This is because in both the treated and the control groups, we have about 17\% of individuals who do not possess any livestock.

\begin{figure}
\centering
\includegraphics[page=2]{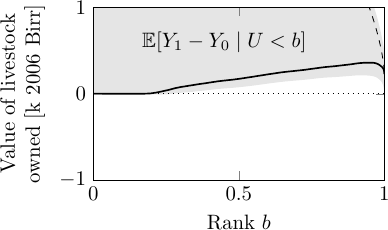}
\caption{Bounds for the STEs on the total value of livestock owned. The black lines indicate the estimated lower bounds for $\mathbb{E}[Y_1-Y_0\mid U<b]$ for each $b$. The gray areas are the pointwise 95\% confidence intervals.}
\label{fig:tdj2015:ste}
\end{figure}

The shaded area presents the pointwise 95\% confidence interval for each $b$, constructed with the method developed by \citet{im2004}.
Although hard to see it in the figure, at $b=1$, the lower end of the confidence interval swiftly crosses $0$, yielding an insignificant ATE estimate.
The lower confidence interval is in fact positive from $b=0.271$ to $b=0.999$.
Thus, unlike the ATE, the STE for those who possess a small to medium amount of livestock is estimated to be significantly positive.

One possible factor behind the insignificant ATE is that the large values of livestock of the ``rich'' were highly volatile.
As is the case for the income or wealth distribution, we see those who have so much that a small multiplicative variation causes high variation in the average.

\subsection{Who Should Be Treated? A Welfare Analysis}

The heterogeneous effects across different levels of $Y_{i0}$ motivate optimizing assignment based on pre\hyp{}treatment $Y_i$.
If we have panel data in which pre\hyp{}treatment outcome is observed and is associated with individuals in the post\hyp{}treatment period, we may simply include it as a covariate and proceed on to maximizing the empirical welfare as in \citet{kt2018}.
If we do not have such data, we might want to find a proxy for the pre\hyp{}treatment outcome.
A natural proxy for the pre\hyp{}treatment outcome is the potential outcome $Y_{i0}$, especially in applications where rapid migration of ranks is considered uncommon.%
\footnote{Another proxy would be a fitted outcome predicted by covariates \citep{acw2016}.}

The data we use consist of two cross\hyp{}sectional surveys, so it is not an option to include the pre\hyp{}treatment outcome at the individual level.
Moreover, in the application where motivation is drawn from ``helping the poor escape the poverty trap,'' uninterventional rank migration is expected to be limited, at least between the poor and the rich.%
\footnote{Note that rank migration {\em within} the subgroup does not affect the value of the STE or its bounds.}
Thus, we consider an optimal treatment assignment based on $Y_{i0}$ that maximizes welfare to be a good policy\hyp{}relevant measure.

We consider an assignment rule in which individuals with values of livestock below some cutoff are granted access to microfinance.
Let $b$ denote the cutoff of assignment in terms of the rank.
Note that we can also state the equivalent assignment rule in terms of the value of livestock, in which case the cutoff is $Q_0(b)$.

We consider two welfare functions.
In the first welfare model, we let the {\em social benefit} equal to the aggregate treatment effect, $\mathbb{E}[(Y_{i1}-Y_{i0})\mathbbm{1}\{U_i<b\}]$, which is normalized to the per\hyp{}capita level.
The {\em social cost}, on the other hand, is set to 100 Birr per individual for the following reasoning.
The average amount of outstanding loans from microfinance for a treated individual was 299 Birr, and we are assuming conservatively that one third of it goes sour.
This is admittedly an assumption too simplistic for a policy analysis, and in reality the default rate would depend on the initial level of $Y$ as well as on many other factors.
However, to concentrate on the key ideas of the paper, we will not seek to find more realistic versions of the social costs.

\begin{figure}
\centering
\includegraphics[page=8]{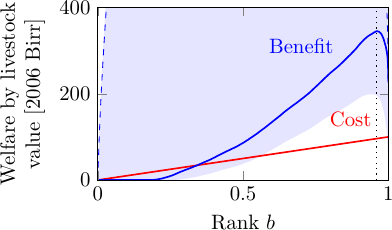}
\caption{Welfare bounds when the treatment is assigned to individuals whose $Y_0$ is below rank $b$. The solid blue line gives the lower bound for the benefits, and the thin dashed blue line gives the upper bound. The red line is the social costs, which is taken to be 100 Birr per treated individual. The blue area is the pointwise 95\% confidence interval for the social benefit. The vertical dotted line ($b=0.96$) is where the worst\hyp{}case welfare is maximized.}
\label{fig:tdj2015:wel}
\end{figure}

In sum, the first social welfare we consider is
\[
	\mathbb{E}[(Y_{i1}-Y_{i0})\mathbbm{1}\{U_i<b\}]-100\cdot\mathbb{E}[\mathbbm{1}\{U_i<b\}].
\]
\Cref{fig:tdj2015:wel} shows the social benefit bounds and social costs as functions of $b$.
The solid blue line is the lower bound for the social benefits (the worst\hyp{}case social benefits), and the dashed blue line is the upper bound.
The solid red line is the social costs.
The shaded blue area is the pointwise 95\% confidence interval for the social benefits.
From $b=0.346$ to $b=1$, the worst\hyp{}case social benefits surpasses the social costs.

To determine the optimal assignment rule, we take the minimax approach: assign the treatment to maximize the worst\hyp{}case welfare.
The worst\hyp{}case welfare corresponds to the lower bound of the welfare.
Thus, the optimal assignment rule is characterized by the cutoff value $b$ that maximizes the difference between the social benefit lower bound and the social cost, which is $b=0.96$.
This is equivalent to providing microfinance to those whose value of livestock is below 9,200 Birr.

Next, we consider a non\hyp{}utilitarian welfare that exhibits loss aversion.
Suppose that the policymaker weighs the loss 10\% more than the benefits in individual $Y_i$, that is, for a function
\[
	h(x)=\begin{cases}x&x\geq 0,\\1.1 x&x<0,\end{cases}
\]
define the second social welfare by
\[
	\mathbb{E}[h(Y_{i1}-Y_{i0})\mathbbm{1}\{U_i<b\}].
\]
Here, we assume that the social costs are already built in to the loss\hyp{}aversive welfare function of the policymaker.

\begin{figure}
\centering
\includegraphics[page=9]{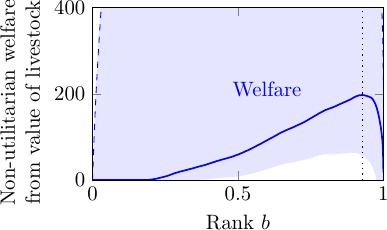}
\caption{Non\hyp{}utilitarian welfare bounds when the treatment is assigned to individuals whose $Y_0$ is below rank $b$. The welfare function weighs the loss 10\% more than the gain. The solid blue line gives the lower bound for the welfare, and the thin dashed blue line gives the upper bound. The blue area is the pointwise 95\% confidence interval for the social welfare. The dotted line ($b=0.928$) is where the worst\hyp{}case welfare is maximized.}
\label{fig:tdj2015:wel2}
\end{figure}

\Cref{fig:tdj2015:wel2} shows the social welfare bounds as functions of $b$.
The solid blue line is the lower bound (the worst\hyp{}case welfare) and the dashed blue line is the upper bound, as characterized by \cref{lem:1}.
The shaded blue area is the pointwise 95\% confidence interval for the social welfare.
From $b=0.357$ to $b=0.976$, the worst\hyp{}case welfare is pointwise significantly positive.

The optimal assignment rule that maximizes the worst\hyp{}case welfare is found to be $b=0.928$, which is equivalent to treating everyone below 7,299 Birr in the total value of livestock owned.

\section{Application 2: SPW for Welfare Reform} \label{sec:labor}

This section presents an empirical application showing how our bounds can be used to provide additional information in the context of testing economic theory.
We estimate the bounds on the proportion of workers who earns more under the new welfare program than the counterfactual under the old program.
The motivation and data are borrowed from \citet{bgh2006}.%
\footnote{The replication files are accessible at the journal's website \citep{bgh2006dta}.}

The summary of the background is as follows.
Aid to Families with Dependent Children (AFDC) was a federal assistance program started in 1935 that required all states to provide financial support to children whose families had low or no income.
During the 1990s, the federal government waived portions of the requirements and allowed states to make changes to various aspects of the program, such as expanding income disregards, increasing work requirements, and introducing time limits on benefits.
In exchange for the waiver, states were required to conduct rigorous evaluations of the impacts of these changes.
Connecticut introduced the Jobs First program as the AFDC waiver, and for the rigorous evaluation, conducted the random assignment study of Jobs First.
The experiment took place between January 1996 and February 1997.
We refer the reader to \citet{bgh2006} for more details.

\citet{bgh2006} stated that
\begin{quote}
. . . theory makes heterogeneous predictions concerning the sign and magnitude of the response of labor supply and welfare use to these reforms. . . . the vast majority of welfare reform studies rely on estimating mean impacts. Theory predicts that these mean impacts will average together positive and negative labor supply responses, possibly obscuring the extent of welfare reform's effects.
\end{quote}
\citet{bgh2006} further used the QTE to demonstrate that the effects were heterogeneous across different levels of the potential outcome variables.
In this section, we complement their analysis by identifying maximum levels of winners and losers given by our bounds.

\begin{figure}
\centering
\begin{subfigure}[t]{0.48\textwidth}
\centering
\includegraphics[page=1]{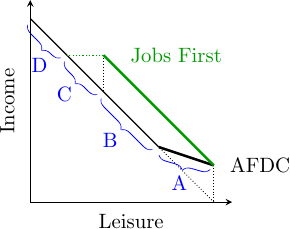}
\caption{Budget constraints before time limit.}
\label{fig:bgh2006:1}
\end{subfigure}
\quad
\begin{subfigure}[t]{0.48\textwidth}
\centering
\includegraphics[page=2]{fig_bgh2006.pdf}
\caption{Budget constraints after time limit.}
\label{fig:bgh2006:2}
\end{subfigure}
\caption{Income-leisure budget constraints under AFDC and under Jobs First \citep[Figure 1]{bgh2006}. Food stamp benefits are not included in these figures.}
\label{fig:bgh2006}
\end{figure}

\Cref{fig:bgh2006} illustrates the stylized budget constraints faced by women supported by AFDC and Jobs First.
The horizontal axis is the time for leisure, which is a complement of working hours, and the vertical axis is the income.
Assuming the standard utility theory, each program participant enjoys the income and leisure on the budget constraint.
Under AFDC, the budget constraint is given by the black solid line in each figure of \Cref{fig:bgh2006}.
If an individual picks a point on segment A, it means that she receives benefits from AFDC.
Jobs First replaced the AFDC benefits with the green solid line, meaningfully raising the budget constraint for individuals on A, B, and C.
The Jobs First program also came with a time limit; after 21 months, the budget constraint is pushed down to just the original black line with no benefits.
Before the time limit (\Cref{fig:bgh2006:1}), individuals on segments A, B, and C would move to somewhere on the green line under Jobs First; individuals on segment D may or may not move to a point on the green line.
After the time limit (\Cref{fig:bgh2006:2}), individuals would receive zero cash transfers and become weakly worse off than AFDC.%
\footnote{They may still have received food stamps.}

While \Cref{fig:bgh2006} provides a helpful picture, the actual implementation was different from the stylized description.
The most notable difference was an extension of the deadline.
In a fair amount of cases, either an extension of 6 months or an indefinite exemption from the time limit was granted.
Despite this, \citet[Figure 2]{bgh2006} provided evidence that the time limit still mattered.
However, it is important to keep in mind that there were beneficiaries that continued to receive benefits after the time limit.

\subsection{How Did the New Program Affect Hours Worked?}

We are interested in whether the new program increases or decreases earnings, which are roughly proportional to the hours worked.

\begin{figure}[t]
\centering
\begin{subfigure}[t]{0.32\textwidth}
\centering
\includegraphics[page=1]{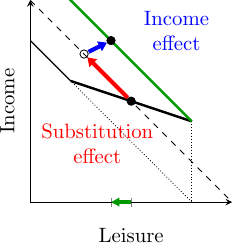}
\caption{When A works more before time limit.}
\label{fig:bgh2006:effects:1}
\end{subfigure}
\ 
\begin{subfigure}[t]{0.32\textwidth}
\centering
\includegraphics[page=2]{fig_bgh2006_effects.pdf}
\caption{When A works less before time limit.}
\label{fig:bgh2006:effects:2}
\end{subfigure}
\ 
\begin{subfigure}[t]{0.32\textwidth}
\centering
\includegraphics[page=3]{fig_bgh2006_effects.pdf}
\caption{When A works more after time limit.}
\label{fig:bgh2006:effects:3}
\end{subfigure}
\caption{Decomposition of the program effect into the Slutsky substitution effect and the Slutsky income effect for the subpopulation on segment A in \Cref{fig:bgh2006}. The total program effect on leisure is the green arrow on the $x$\hyp{}axis.}
\label{fig:bgh2006:effects}
\end{figure}

There are three subpopulations of interest.
We assume that each individual faces the same wage under the new and old programs, one of which is counterfactual.
We do {\em not} assume that the wage is common across individuals.%
\footnote{However, we can only measure hours worked by earnings, so SPWs are conditional on earnings.}
Under the new program, the following are the predictions of the stylized theory.
\begin{enumerate}[i.]
	\item The subpopulation on segment A in \Cref{fig:bgh2006:1} might work more or work less, depending on the magnitudes of the substitution effect (increases hours worked) and the income effect (decreases hours worked). This is shown in \Cref{fig:bgh2006:effects:1,fig:bgh2006:effects:2}.
		\label{q1}
	\item The subpopulation on segment D in \Cref{fig:bgh2006:1} might work less or work as much.
		\label{q2}
	\item The subpopulation on segment A in \Cref{fig:bgh2006:2} would work more if leisure is a normal good. This is shown in \Cref{fig:bgh2006:effects:3}.
		\label{q3}
\end{enumerate}

\begin{figure}[t]
\centering
\begin{subfigure}[t]{0.32\textwidth}
\centering
\includegraphics[page=2]{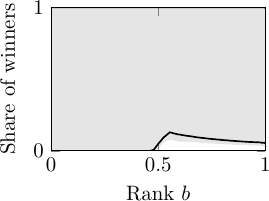}
\caption{SPW for $\{U<b\}$ in Q4.}
\label{fig:bgh2006:ernq4:L}
\end{subfigure}
\ 
\begin{subfigure}[t]{0.32\textwidth}
\centering
\includegraphics[page=3]{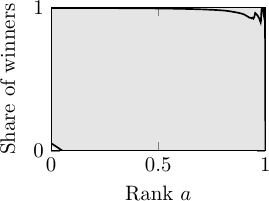}
\caption{SPW for $\{U>a\}$ in Q4.}
\label{fig:bgh2006:ernq4:R}
\end{subfigure}
\ 
\begin{subfigure}[t]{0.32\textwidth}
\centering
\includegraphics[page=2]{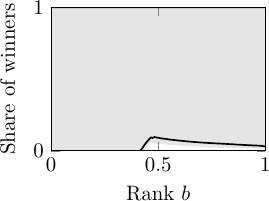}
\caption{SPW for $\{U<b\}$ in Q12.}
\label{fig:bgh2006:ernq12:L}
\end{subfigure}
\caption{Estimated bounds on the SPW in earnings in post\hyp{}treatment quarters 4 and 12. A small fraction of significant winners are found on the left in Q4 and Q12.}
\label{fig:bgh2006:ernq}
\end{figure}

First, we examine the sign of the program effects on subpopulation (\ref{q1}).
For this, we plot the estimated bounds on the SPW on earnings in the 4th post\hyp{}treatment quarter (Q4) for the subgroup defined by $\{U_i<b\}$ as a function of $b$ in \Cref{fig:bgh2006:ernq4:L}.
The black line on the bottom that is zero until $b=0.47$ and takes a nonzero value thereafter is the estimated lower bound of the SPW
\[
	P(Y_{i1}>Y_{i0}\mid U_i<b),
\]
where $Y_i$ is $i$'s earnings.
The black line that is almost identical to one is the upper bound.
These are the median\hyp{}bias\hyp{}corrected estimators given by \citet{clr2013}.
For this SPW, the upper bound is uninformative, while the lower bound identifies some winners who increased hours worked under the new program.
Namely, the lower bound takes a maximum of $0.13$ at $b=0.55$.
This means that, among the workers whose earnings are below the 55\% quantile under the old program, at least 13\% of them worked more under the new program.
The gray area is the pointwise 95\% confidence interval constructed by the method of \citet{clr2013}.

These numbers are modest, but seem coherent with the QTE estimates of \citet[Figure 3]{bgh2006}.
Considering that the theoretical prediction was ambiguous, the bounds successfully identified a small yet significant portion of individuals whom the new program encouraged to work more.
We did not identify losers, who worked less because of the new program.
The lack of evidence does not mean the lack of losers, but the data was not decisive enough to say anything about their existence or prevalence.
In \cref{sec:labor:tighten}, we will see that an additional assumption can sometimes tighten the bounds.

Second, we examine whether the possible negative effect on subpopulation (\ref{q2}) can be identified in the data.
\Cref{fig:bgh2006:ernq4:R} plots the estimated bounds on the SPW for $\{U_i>a\}$ as a function of $a$.
As before, the black line at the bottom is the lower bound, and the one at the top is the upper bound.
The lower bound is zero between $a=0.055$ and $1$.
The upper bound is slightly off $1$, but the confidence interval shows that it is not significantly away from $1$ at any value of $a$.
In short, we do not see a statistically significant amount of losers on the right tail of the earnings distribution.

Third, after the time limit, the new program stopped providing benefits to some individuals, which was expected to encourage them to work more for subpopulation (\ref{q3}).
We examine this effect in \Cref{fig:bgh2006:ernq12:L}.
It plots the bounds on the SPW on the earnings in the 12th post\hyp{}treatment quarter (Q12) for $\{U_i<b\}$ as a function of $b$.%
\footnote{By the 12th quarter, those who received a 6\hyp{}month extension would have stopped receiving benefits.}
The lower bound becomes positive after $b=0.42$ and takes the maximum of $0.097$ at $b=0.48$.
Therefore, among those whose earnings are below the 48\% quantile under the old program, at least 9.7\% of them work more under the new program.
The pointwise 95\% confidence interval leaves zero at $b=0.44$.
The fact that the bounds identify significant winners at lower values of $b$ compared to \Cref{fig:bgh2006:ernq4:L} is consistent with the observation that if you receive low to zero earnings, the anticipated program effect on the hours worked would be large, due to both the positive substitution effect and the positive income effect (\Cref{fig:bgh2006:effects:3}).

\subsection{Using Economic Theory to Tighten the Bounds} \label{sec:labor:tighten}

\begin{figure}[t]
\centering
\begin{subfigure}[t]{0.48\textwidth}
\centering
\includegraphics[page=1]{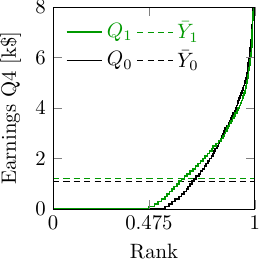}
\caption{Quantile functions.}
\label{fig:bgh2006:ernq4:Q}
\end{subfigure}
\quad
\begin{subfigure}[t]{0.48\textwidth}
\centering
\includegraphics[page=2]{fig_joint.pdf}
\caption{Domain restriction.}
\label{fig:bgh2006:ernq4:DR}
\end{subfigure}
\caption{Restriction imposed by utility theory. The lower 47.5\% of $Y_1$ is zero, which is matched with the lower 47.5\% of $Y_0$. As a result, the joint distribution is restricted to be supported on the shaded areas in (\subref{fig:bgh2006:ernq4:DR}).}
\label{fig:bgh2006:ernq:joint}
\end{figure}

Our bounds are nonparametrically tight, meaning that they characterize the maximum amounts of winners and losers without having to make additional assumptions.
However, it is also true that the estimated bounds in \Cref{fig:bgh2006:ernq} were only modestly informative.

In this section, we investigate how we can tighten the bounds if we are open to making more assumptions.
Recall that the theoretical prediction we examined came from a stylized utility theory.
If we take the utility theory for granted, we see that those who do not work under the new program would not work under the old program if leisure is a normal good (see \Cref{fig:bgh2006}).%
\footnote{The leisure being a normal good is only needed for after the time limit. Also, the converse does not hold in that those who do not work under the old program may work nonzero hours as discussed for subpopulation (\ref{q3}).}

This is in line with the data; in the 4th post\hyp{}treatment quarter, 47.5\% do not work under the new program while 55.4\% do not under the old program (\Cref{fig:bgh2006:ernq4:Q}).
Thus, we may assume that those 47.5\% in the treated go straight into the 55.4\% in the control.
The difference 7.9\% are those who work under the new program but do not under the old.
In light of this, we eliminate the 47.5\% from the treated sample as well as the 47.5\% off of the 55.4\% from the control sample.
This leaves us with a smaller sample, with which we compute the bounds as before.
This is equivalent to imposing a block\hyp{}diagonal support restriction on the joint distribution of the potential outcomes (\Cref{fig:bgh2006:ernq4:DR}).%
\footnote{\Cref{fig:bgh2006:ernq4:DR} can be seen as a restriction on the support of the copula density.}

The new bounds estimated with the refined samples are presented in \Cref{fig:bgh2006:ernq4:tight:L}.
The horizontal axis now starts from $b=0.475$ since we have eliminated the 47.5\% of non\hyp{}working individuals in both groups.
Everyone represented in this figure works under the new program.
Both the upper and lower bounds are one for the first 7.9\%.
Then, the lower bound decreases as we expand the subgroup.
This sharp decline suggests that we hardly identify any more winners than the first 7.9\%.

Meanwhile, the domain restriction identifies a small portion of significant losers on the upper tail (\Cref{fig:bgh2006:ernq4:tight:R}).
The upper end of the pointwise 95\% confidence interval is strictly below $1$ between $a=0$ and $a=0.90$.
Within this range, the estimated upper bound takes the minimum of $0.879$ at $a=0.859$.
This indicates that, among the subpopulation whose $Y_0$ is above \$4,500, at least 12.1\% of them work fewer hours under the new program.

In the 12th quarter, 46.6\% in the control did not work, and neither did 41.6\% in the treated.
Applying the same elimination procedure, the estimated bounds on the SPW for the subgroup $\{U_i<b\}$ are given in \Cref{fig:bgh2006:ernq12:tight:L}.
Like \Cref{fig:bgh2006:ernq4:tight:L}, we do not identify more winners than the initial 5.0\% who do not work under the old program.

\begin{figure}[t]
\centering
\begin{subfigure}[t]{0.32\textwidth}
\centering
\includegraphics[page=2]{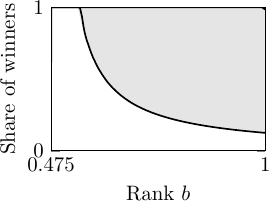}
\caption{SPW for $\{U<b\}$ in Q4.}
\label{fig:bgh2006:ernq4:tight:L}
\end{subfigure}
\ 
\begin{subfigure}[t]{0.32\textwidth}
\centering
\includegraphics[page=3]{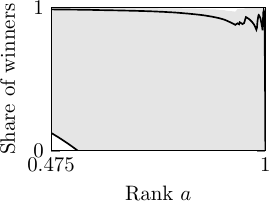}
\caption{SPW for $\{U>a\}$ in Q4.}
\label{fig:bgh2006:ernq4:tight:R}
\end{subfigure}
\ 
\begin{subfigure}[t]{0.32\textwidth}
\centering
\includegraphics[page=2]{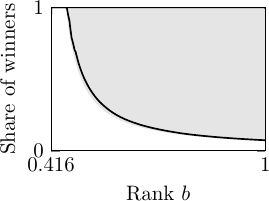}
\caption{SPW for $\{U<b\}$ in Q12.}
\label{fig:bgh2006:ernq12:tight:L}
\end{subfigure}
\caption{Tightened bounds on the SPW in earnings in Q4 and Q12 assuming utility theory.}
\label{fig:bgh2006:ernq4:tight}
\end{figure}

\section{Conclusion} \label{sec:concl}

Many empirical questions concern treatment effects at different levels of the outcome variables.
We developed two bounds to investigate heterogeneous treatment effects of this type.
These bounds can be computed with the estimated marginal distributions of the potential outcomes, and provide complementary information to the widely used QTE.
They are sharp without requiring additional assumptions.

The first bounds are on the STE of the form $\mathbb{E}[Y_{i1}-Y_{i0}\mid Y_{i0}<c]$.
In \cref{sec:micro}, we applied them to assess the effect of microfinance on the poor, and found that, despite the ATE being insignificant, the STE for various values of $c$ was significantly positive.
We also applied our bounds to the policy targeting problem.
We considered two measures of welfare: one linear in the treatment effect with per\hyp{}capita fixed costs and the other a non\hyp{}utilitarian welfare that embodied a loss\hyp{}aversive preference.
We illustrated how our bounds led to treatment assignment rules that maximized the worst\hyp{}case welfare for each welfare measure.

The second bounds are on the SPW of the form $P(Y_{i1}>Y_{i0}\mid Y_{i0}<c)$.
In \cref{sec:labor}, we investigated the heterogeneous effects of Connecticut's welfare reform on the earnings predicted by stylized labor supply theory.
For a subgroup for whom the theoretical prediction was ambiguous, we detected a portion of individuals who increased hours worked because of the new welfare program.
We also illustrated how economic theory could be used to tighten the bounds, and after tightening, we detected a small portion of individuals who decreased hours worked because of the new welfare program.

We conclude by noting that our results are applicable well beyond the RCT setups picked up in our empirical applications.
The bounds are based solely on the marginal distributions of the potential outcomes, which most QTE estimation procedures estimate as byproducts.
For example, there are existing QTE estimators for the RCT with imperfect compliance, difference\hyp{}in\hyp{}differences, selection\hyp{}on\hyp{}observables, regression discontinuity designs, and instrumental variables regression models.
Our bounds provide interpretable assessment of the heterogeneity of treatment effects that supplements the QTE.
The Stata program for our bounds can be downloaded at an \href{https://github.com/jcao0/subgroup-treatment-effects}{author's website}.

\newpage

\begin{appendices}
\appendix
\appendixpage

\section{Proofs} \label{sec:theory}

\cref{thm:ste} is an obvious corollary of \cref{lem:1}.
To prove \cref{lem:1}, we first obtain the second\hyp{}order stochastic dominance bounds on the subgroup distribution of treatment effects.
This extends the existing results on the entire distribution of treatment effects \citep[Lemma 2.2]{fp2010}.

\begin{thm}[Second\hyp{}order stochastic dominance bounds on subgroup distribution of treatment effects] \label{thm:sosd}
Let $0\leq a<b\leq 1$ and $U\sim U(0,1)$ such that $Y_0=Q_0(U)$.
Conditional on $a<U<b$, we have
\[
	Q_1(b-U)-Q_0(U)\preceq_{2}Y_1-Y_0\preceq_{2}Q_1(1-b+U)-Q_0(U),
\]
where $A \preceq_{2} B$ means that the distribution of $A$ is second\hyp{}order stochastically dominated by that of $B$.
\end{thm}

\begin{proof}%
Let $h:\mathbb{R}\to\mathbb{R}$ be a nondecreasing concave function.
Consider the optimal transport from $y\in(a,b)$ to $x\in(0,b-a)$ with cost $h(Q_1(x)-Q_0(y))$.
Then, \citet[Theorem 3.1.2]{rr1998} imply
\[
	\mathbb{E}[h(Q_1(b-V_0)-Q_0(V_0))]\leq\mathbb{E}[h(Q_1(V_1)-Q_0(V_0))]
\]
where $V_0\sim U(a,b)$ and $V_1\sim U(0,b-a)$.
Then, \citet[Theorem 3]{bgms2009} imply that there exist functions $\phi$ and $\psi$ such that for $x\in(0,b-a)$ and $y\in(a,b)$, we have
\(
	\phi(x)+\psi(y)\leq h(Q_1(x)-Q_0(y))
\)
with equality when $x=b-y$.%
\footnote{\citet[Theorem 3]{bgms2009} assume nonnegative costs, but the theorem trivially extends to costs bounded from below. Then optimality follows for cost $h(Q_1(\cdot)-Q_0(\cdot))\vee A$ for every $A\in\mathbb{R}$, so let $A\to-\infty$.}
Here, $\phi$ is nondecreasing since for every $y'\geq y$,
\[
	\phi(b-y)+\psi(y)=h(Q_1(b-y)-Q_0(y))\geq h(Q_1(b-y')-Q_0(y))\geq\phi(b-y')+\psi(y).
\]
Thus, we see that for $x\in(0,1)$ and $y\in(0,1)$,
\[
	\phi(x\wedge(b-a))+\psi(y)\mathbbm{1}\{y\in(a,b)\}-\phi(b-a)\mathbbm{1}\{y\notin(a,b)\}\leq h(Q_1(x)-Q_0(y))\mathbbm{1}\{y\in(a,b)\}
\]
with equality when $x=b-y$ and $y\in(a,b)$ or when $x\geq b-a$ and $y\notin(a,b)$.
Thus, \citet[Theorem 3]{bgms2009} imply
\[
	\mathbb{E}[h(Q_1(b-U)-Q_0(U))\mathbbm{1}\{a<U<b\}]\leq\mathbb{E}[h(Y_1-Y_0)\mathbbm{1}\{a<U<b\}],
\]
that is, $Q_1(b-U)-Q_0(U)\preceq_{2}Y_1-Y_0$ conditional on $a<U<b$.

Similarly, \citet[Theorem 3.1.2]{rr1998} imply
\[
	\mathbb{E}[h(Q_1(V_1)-Q_0(V_0))]\leq\mathbb{E}[h(Q_1(1-b+V_0)-Q_0(V_0))]
\]
where $V_0\sim U(a,b)$ and $V_1\sim U(1-b+a,1)$.
\citet[Theorem 3]{bgms2009} imply that there exist functions $\phi$ and $\psi$ such that for $x\in(1-b+a,1)$ and $y\in(a,b)$, we have
\(
	h(Q_1(x)-Q_0(y))\leq\phi(x)+\psi(y)
\)
with equality when $x=1-b+y$.
Here, $\phi$ is nondecreasing since for every $y'\geq y$,
\begin{align*}
	\phi(1-b+y)+\psi(y)&=h(Q_1(1-b+y)-Q_0(y))\\
	&\leq h(Q_1(1-b+y')-Q_0(y))\leq\phi(1-b+y')+\psi(y).
\end{align*}
Thus, we see that for $x\in(0,1)$ and $y\in(0,1)$,
\begin{multline*}
	h(Q_1(x)-Q_0(y))\mathbbm{1}\{y\in(a,b)\}\\
	\leq\phi(x\vee(1-b+a))+\psi(y)\mathbbm{1}\{y\in(a,b)\}-\phi(1-b+a)\mathbbm{1}\{y\notin(a,b)\}
\end{multline*}
with equality when $x=1-b+y$ and $y\in(a,b)$ or when $x\leq 1-b+a$ and $y\notin(a,b)$.
Thus, \citet[Theorem 3]{bgms2009} imply
\[
	\mathbb{E}[h(Y_1-Y_0)\mathbbm{1}\{a<U<b\}]\leq\mathbb{E}[h(Q_1(1-b+U)-Q_0(U))\mathbbm{1}\{a<U<b\}],
\]
that is, $Y_1-Y_0\preceq_{2}Q_1(1-b+U)-Q_0(U)$ conditional on $a<U<b$.
\end{proof}

\cref{lem:1} turns out to be equivalent to \cref{thm:sosd}.

\begin{proof}[Proof of \cref{lem:1}]
Let $\tilde{f}(x)=-f(-x)$, $\tilde{U}=1-U$, and $\tilde{Y}_j=-Y_j$ so the quantile function of $\tilde{Y}_j$ is $\tilde{Q}_j(u)=-Q_j(1-u)$.
Since $\tilde{f}$ is nondecreasing and concave, \cref{thm:sosd} implies
\begin{multline*}\!\!\!\!\!
	\mathbb{E}[\tilde{f}(\tilde{Q}_1(1-a-\tilde{U})-\tilde{Q}_0(\tilde{U}))\mathbbm{1}\{1-b<\tilde{U}<1-a\}]\leq\mathbb{E}[\tilde{f}(\tilde{Y}_1-\tilde{Y}_0)\mathbbm{1}\{1-b<\tilde{U}<1-a\}]\\
	\leq\mathbb{E}[\tilde{f}(\tilde{Q}_1(a+\tilde{U})-\tilde{Q}_0(\tilde{U}))\mathbbm{1}\{1-b<\tilde{U}<1-a\}],
\end{multline*}
which reduces to the first claim.
Since $-g$ is nondecreasing and concave, the second claim follows straightforwardly from \cref{thm:sosd}.
\end{proof}

The next theorem bounds the unconditional probabilities of winners and losers.
\cref{thm:spw} can be obtained by dividing all sides in \cref{thm:winner} by $b-a$.

\begin{thm}[Bounds on subgroup proportions of winners and losers] \label{thm:winner}
For every $0\leq a<b\leq 1$,
\begin{multline*}
	\sup_{u\in(a,b)}[u-a-F_1(Q_0(u))]_+\leq P(Y_1>Y_0,a<U<b)\\
	\shoveright{\leq b-a-\sup_{u\in(a,b)}[b-u-1+F_1(Q_0(u))]_+,}\\
	\shoveleft{\sup_{u\in(a,b)}[b-u-1+F_1(Q_0(u)-)]_+\leq P(Y_1<Y_0,a<U<b)}\\
	\leq b-a-\sup_{u\in(a,b)}[u-a-F_1(Q_0(u)-)]_+.
\end{multline*}
For each of the lower bounds, there exists a joint distribution of $(Y_0,Y_1)$ that attains it; for each of the upper bounds, there exists a joint distribution of $(Y_0,Y_1)$ that makes it arbitrarily tight.
All bounds are nonincreasing and Lipschitz in $a$ and nondecreasing and Lipschitz in $b$.
\end{thm}

\begin{proof}[Proof of \cref{thm:winner}]
For the first lower bound, observe that
\begin{multline*}
	P(a<U<b,Q_0(U)<Y_1)
	\geq\sup_{u\in(a,b)}P(a<U<u,Q_0(u)<Y_1)\\
	\geq\sup_{u\in(a,b)}[P(a<U<u)+P(Q_0(u)<Y_1)-1]_+
	=\sup_{u\in(a,b)}[u-a-F_1(Q_0(u))]_+,
\end{multline*}
where the second inequality uses the lower Fr\'echet inequality.%
\footnote{We thank Ismael Mourifi\'e for suggesting the use of the Fr\'echet inequality in this proof.}
To show tightness, let $\Delta=\sup_{u\in(a,b)}[u-F_1(Q_0(u))]_+$ and consider $Y_1=Q_1(U_1)$ such that $U_1=U-a-\Delta$ if $a+\Delta<U<b$.
The joint distribution of $(U_1,U)$ when $U\notin(a+\Delta,b)$ is arbitrary.
Then $U=u\in(a+\Delta,b)$ are not winners since
\[
	Q_1(u-a-\Delta)-Q_0(u)\leq Q_1(u-a-\Delta)-Q_1(F_1(Q_0(u)))\leq 0,
\]
where the last inequality follows from the monotonicity of $Q_1$ and
\[
	(u-a-\Delta)-F_1(Q_0(u))=[u-a-F_1(Q_0(u))]-\Delta\leq 0.
\]
This means that the lower bound is binding for this joint distribution.

To derive the first upper bound, observe that
\begin{align*}
	P(a<U<b,Q_0(U)<Y_1)
	&=P(a<U<b)-P(a<U<b,Y_1\leq Q_0(U))\\
	&\leq b-a-\sup_{u\in(a,b)}[b-u-1+F_1(Q_0(u))]_+,
\end{align*}
where the inequality follows as the lower bound.
For tightness, let $\Delta=b-a-\sup_{u\in(a,b)}[b-u-1+F_1(Q_0(u))]_+$ and, for arbitrary $\varepsilon>0$, consider $Y_1=Q_1(U_1)$ such that $U_1=U-a+1-\Delta+\varepsilon$ if $a<U<a+\Delta-\varepsilon$.
The joint distribution of $(U_1,U)$ for $U\notin(a,a+\Delta-\varepsilon)$ is arbitrary.
Then $U=u\in(a,a+\Delta-\varepsilon)$ are winners since
\begin{multline*}
	F_1(Q_1(u-a+1-\Delta+\varepsilon))-F_1(Q_0(u))\geq u-a+1-\Delta+\varepsilon-F_1(Q_0(u))\\
	=b-a-[b-u-1+F_1(Q_0(u))]-\Delta+\varepsilon\geq\varepsilon>0,
\end{multline*}
which implies $Q_1(u-a+1-\Delta+\varepsilon)-Q_0(u)>0$.
Since this holds for arbitrarily small $\varepsilon>0$, the upper bound can be arbitrarily tight.

The lower bound is nondecreasing and Lipschitz in $b$ since for every $\varepsilon>0$,
\begin{multline*}
	\sup_{u\in(a,b)}[u-a-F_1(Q_0(u))]_+\leq\sup_{u\in(a,b+\varepsilon)}[u-a-F_1(Q_0(u))]_+\\
	\leq\sup_{u\in(a,b+\varepsilon)}[u-a-F_1(Q_0(u\wedge b))]_+\leq\sup_{u\in(a,b)}[u-a-F_1(Q_0(u))]_++\varepsilon.
\end{multline*}
The lower bound is nonincreasing in $a$ since for every $\varepsilon>0$,
\[
	\sup_{u\in(a-\varepsilon,b)}[u-a+\varepsilon-F_1(Q_0(u))]_+
	\geq\!\!\!\sup_{u\in(a-\varepsilon,b)}[u-a-F_1(Q_0(u))]_+
	\geq\!\sup_{u\in(a,b)}[u-a-F_1(Q_0(u))]_+
\]
and is Lipschitz in $a$ since for every $\varepsilon>0$,
\begin{multline*}\!\!\!\!\!\!
	\sup_{u\in(a-\varepsilon,b)}[u-a+\varepsilon-F_1(Q_0(u))]_+
	\leq\sup_{u\in(a-\varepsilon,a]}(u-a+\varepsilon)_+\vee\sup_{u\in(a,b)}[u-a+\varepsilon-F_1(Q_0(u))]_+\\
	\leq\varepsilon\vee\biggl(\sup_{u\in(a,b)}[u-a-F_1(Q_0(u))]_++\varepsilon\biggr)
	=\sup_{u\in(a,b)}[u-a-F_1(Q_0(u))]_++\varepsilon.
\end{multline*}
The same properties follow analogously for the upper bound.

The loser bounds can be derived analogously.
\end{proof}

\end{appendices}

\bibliographystyle{econ}
\bibliography{stebib}
\addcontentsline{toc}{section}{\refname}

\end{document}